\documentclass[12pt]{article}
\RequirePackage[OT1]{fontenc}
\RequirePackage[utf8]{inputenc}
\RequirePackage{amsthm,amsmath,amssymb,amsfonts}
\RequirePackage[numbers]{natbib}
\RequirePackage[colorlinks,citecolor=blue,urlcolor=blue]{hyperref}
\RequirePackage{xfrac}
\usepackage{authblk}
\usepackage{setspace}
\usepackage{fullpage}
\onehalfspace
\newtheorem{theorem}{Theorem}
\newtheorem*{theorem*}{Theorem}
\newtheorem*{definition}{Definition}
\newtheorem*{assumption}{Assumption}

\newtheorem*{example*}{Example}

\newcommand{\R}{{\mathbb{R}}}
\newcommand{\bs}{\boldsymbol}
\newcommand{\E}{\mathbb{E}}

\def\Z{{\bf Z}}
\def\argmin{\operatorname*{\mathsf{arg\,min}}}
\def\argmax{\operatorname*{\mathsf{arg\,max}}}

\def\U{\mathbb{U}}
\def\V{\mathbb{V}}
\def\M{\mathcal{M}}

\def\Mbar{\bar{\mathcal{M}}}
\def\Mcoll{\mathfrak{M}}
\def\Xfrac{\mathfrak{X}}
\def\Rnorm{\Phi}%\mathcal{R}}
\def\L{\mathcal{L}}

\def\Row{\text{row}}
\def\Col{\text{col}}
\def\Rank{\text{rank}}
\def\Dim{\text{dim}}

\def\McollRt{\Mcoll_{\psi_n^2}}

\def\diag{\operatorname*{diag}}
\def\trace{\operatorname*{trace}}

\newcommand{\vertiii}[1]{{\left\vert\kern-0.25ex\left\vert\kern-0.25ex\left\vert #1 
    \right\vert\kern-0.25ex\right\vert\kern-0.25ex\right\vert}}

\newcommand{\X}{\mathbf{X}}

\title{Generalized Information Criteria for Structured Sparse Models}
\author[1]{Eduardo F. Mendes\thanks{Corresponding author: eduardo.mendes@fgv.br}}
\author[2]{Gabriel J. P. Pinto}
\affil[1]{São Paulo School of Economics, Fundação Getulio Vargas}
\affil[2]{Department of Statistics, Northwestern University}
\date{}                     %% if you don't need date to appear
\begin{document}
\maketitle
\begin{abstract}%
    Regularized $m$-estimators are widely used due to their ability of recovering a low-dimensional model in high-dimensional scenarios. Some recent efforts on this subject focused on creating a unified framework for establishing oracle bounds, and deriving conditions for support recovery. 
    Under this same framework, we propose a new Generalized Information Criteria (GIC) that takes into consideration the sparsity pattern one wishes to recover. 
    We obtain non-asymptotic model selection bounds and sufficient conditions for model selection consistency of the GIC. 
    Furthermore, we show that the GIC can also be used for selecting the regularization parameter within a regularized $m$-estimation framework, which allows practical use of the GIC for model selection in high-dimensional scenarios. 
    We provide examples of group LASSO in the context of generalized linear regression and low rank matrix regression.

\end{abstract}

\section{Introduction}

\subsection{Overview} 
I this paper we derive a family of Generalized Information Criteria tailored for \emph{structured sparse} models and show its statistical properties such as finite sample selection bounds and path-consistency. \emph{Sparsity} is usually related to high dimensional models, but one can also think about it in low dimension settings. In fact, model selection is frequently related to selecting a sparse model in some sense: some of the variables, or groups of variables, are not included or we estimate the number of factors on a factor regression model. What we learn from high-dimensional statistics literature is that the sparsity pattern is often connected to regularizing norms, such as $L_1$ for variable selection in linear regressions, $L_{2,1}$ for selecting groups of variables and the nuclear norm for shrinking eigenvalues of a low rank matrix. These three examples illustrate how sparsity can be structured in a sense that we have deterministic information about the model that is embeded in the collection of models we are choosing from. 

Model selection in high dimensions poses a challenge as the number of candidate models is very large. In these cases one indeed uses structured sparse norms to help estimating models, and controls the sparsity using a regularization parameter, i.e., indexes the candidate models along a path. We call path consistency the property that there exists a regularization parameter that indeed corresponds to teh target model. We show how to construct this sequence and that is indeed contains the target model. As we combine both cases we indeed can use the GIC to select the regularization parameter in high dimensions, meaning that our results are valid in a broad sense.

%But why one should care about structured sparse models at all?
We study both asymptotic and non-asymptotic properties of generalized information criteria when the model has a structured sparsity pattern. The method contemplates situations where the ambient dimension of the problem $p$ is larger, in order, than the sample size $n$. In such cases, we will select the regularization parameter of a regularized $m$-estimation problem. This class of problems is the concern of high-dimensional statistics and frequently arises in economics and finance, as well as astronomy and biology, among other areas. 

Estimating high-dimensional sparse models is infeasible unless one imposes low-dimensional constraints on the model space. The most popular example is variable selection within a linear regression problem where only a small number (${s\ll p}$) of variables are relevant. In this case, sparsity is unstructured as any variable can individually enter or not the model. In structured sparsity there is a natural structure to the problem. For instance, in a group sparse regression, we select variables in groups, whereas in low rank matrix estimation individual elements are not necessarily sparse, even if its singular values are. These examples illustrate what we call \emph{structured sparsity}.

We estimate model parameters by minimizing some empirical loss function that incorporates a structured sparsity penalty. As such, we are within the regularized $m$-estimation framework. Our goal is to select an adequate model subspace from a large collection of candidate models, satisfying sparsity constraints. 

For a comprehensive exposition of the area of high dimensional statistics, including methods, algorithms and theory, see the books \cite{buhlmann2011statistics,tibshirani2015statistical} and, for the non-asymptotic perspective adopted in the paper, see \cite{wainwright2019High}.

\subsection{Literature review}
The Akaike Information Criterion (AIC) \citep{AkaikeAIC}, cross-validation  \citep{stone1974,cravenwhaba1978}, and Schwarz Bayesian Information Criterion (BIC) \citep{schwarz1978} are traditional model selection methods. However, they are not consistent in selecting large model spaces \cite{bromanspeed2002,casellaetal2009,wangetal2009,zhangetal2010,shetran2019}. Several methods have been proposed to address this limitation in high-dimensional settings.

For instance, \citet{chenchen2008,chenchen2012} argue that BIC's uniform prior on candidate models causes over-selection and propose the Extended BIC (EBIC), which uses an alternative prior distribution on model subspaces for variable selection. \citet{wangetal2009} suggest adding a penalty term to BIC to account for candidate model inflation, which leads to the Generalized Information Criterion (GIC). Some authors propose alternative methods, such as \citeauthor{gaosong2010}'s \citep{gaosong2010} Composite Likelihood BIC, \citeauthor{zhangshen2010}'s \citep{zhangshen2010} Corrected Risk Inflation Criterion, and \citeauthor{shetran2019}'s \citep{shetran2019} new cross-validation scheme for sparse reduced rank regression that achieves optimal rate.

\citet{kimjeon2016} extend previous results beyond the quadratic loss or negative log-likelihood and demonstrate that GIC consistently selects the correct model for a broad range of loss functions.

In practice, it is often impossible to enumerate all candidate model subspaces in high-dimensional settings. Regularized $m$-estimators construct a path of candidate models indexed by a tuning parameter or regularization parameter, which may depend on the model space's complexity and sample size. These include LASSO or basis pursuit \cite{tibshirani1996,chenetal2001}, SCAD \cite{fan2001variable}, adaptive LASSO \cite{zou2006}, group LASSO \cite{yuanlin2006}, graphical LASSO for graphs and inverse covariances \cite{nMpB2006,yuanlin2007,friedmanetal2008}, among others.

Numerous studies propose consistent model selection criteria for the regularization parameter. A method is path-consistent if the solution path indexed by the regularization parameter contains the true model, and the selection criteria are consistent. Most often, the proposed criteria use either the quadratic loss or Kullback-Leibler loss (negative log-likelihood function) with LASSO, adaptive LASSO, or SCAD regularization functions. \citet{chenchen2008,chenchen2012} use the EBIC within a regularized maximum likelihood framework with SCAD and LASSO penalties. \citet{wangetal2007,wangetal2009} show that their GIC consistently selects the correct model using the quadratic loss and LASSO or SCAD penalties.\citet{zhangshen2010} propose the high-dimensional BIC (HBIC) for selecting the regularization parameter in sparse regression problems with quadratic loss and SCAD penalty.

Other studies focus on specific families of models. For example, \citet{nicai2018} study selecting the tuning parameter in a penalized Cox proportional hazard model with SCAD penalty, \citet{foygeldrton2010} derive an EBIC for Gaussian graphical models, and \citet{leeetal2014} derive an EBIC for penalized quantile regression. \citet{kimjeon2016} consider a broader class of loss functions called quadratically supported risk and show path-consistency of their GIC for the LASSO and SCAD under a penalized least squares framework.

\subsection{Contribution}
Current high-dimensional model selection methods do not take into account the \emph{sparsity pattern} of the problem. In this paper we develop a formulation of the Generalized Information Criterion that accommodates the sparsity pattern, develop its non-asymptotic selection bound, and demonstrate its consistency and path-consistency for recovering the sparsity pattern. We apply our approach to group sparse GLM estimation and low rank matrix estimation, which have been studied in the literature, and connect our assumptions to those used to establish oracle bounds.

As motivation, consider three matrices with distinct sparsity constraints: (1) low rank, (2) element-wise sparsity, and (3) group row sparsity. In (1) the entries of the matrix may not be sparse in the traditional element-wise way, whereas in (2) the matrix may be full rank. Similarly, in group row sparsity (3) we shrink entire rows of the matrix while in element-wise sparsity (1) individual elements will be null within the same row. Information about the \textit{sparsity pattern} is included in our estimation problem through an appropriate regularization norm. Even though the same loss used in all problems, the regularizing penalties are quite different. In (1) we use the nuclear norm, in (2) the element-wise $L_1$ norm, and in (3) the $L_1/L_2$ row norm, often used in group-LASSO. 

In this paper we address the following questions: (a) how should one incorporate the \emph{sparsity pattern} in the model selection criterion? (b) can this information criterion be used to to select the regularization parameter $\lambda$? (c) what are the finite sample and asymptotic behaviors of the method?

\citet{sNpRmWbY2012} introduce a unified framework for establishing properties of regularized $m$-estimators in high-dimensions. The authors identify two key properties that guide convergence rates: \emph{restricted strong convexity} and \emph{decomposability}. The restricted strong convexity depends on the interaction of the regularizer with the loss function. The decomposability property of the regularizer requires the regularizing norm to be additive in a pair of of orthogonal subspaces. Authors derive sharp bounds on the error norm under these two key properties and other rate restrictions.\footnote{Traditionally $L_2$ for vectors and Frobenius norm for matrices.} Results within this framework are restricted to regularization functions that are norms, such as the LASSO, Group LASSO, low lank recovery, among other, and does not include the SCAD and Elastic Net penalties.

Our formulation of the Generalized Information Criterion accounts for the sparsity pattern in the model. We define a penalty that depends on a \emph{subspace compatibility constant} and some decreasing sequence that is related to the regularization norm used to capture the sparsity. We demonstrate that our proposed GIC is model selection consistent and also path-consistent for selecting the regularization parameter $\lambda$. Our conditions are connected to those in \citet{kimjeon2016} and \citet{sNpRmWbY2012}, and are readily satisfied by many loss functions and penalties in the high-dimensional statistics literature \citep{wainwright2019High}. 

To illustrate the generality of our approach, consider a regularized generalized linear regression setting with LASSO or group LASSO penalty. In the former case, the penalty on the GIC is proportional to $\frac{\log p}{n}s$, where $s$ is the number of non-zero variables in the model, while in the latter case, the penalty function is proportional to $(\frac{m}{n} + \frac{\log g}{n})s_g$, where $s_g$ is the number of selected groups and $m$ the group size. The penalty function in the GIC for group LASSO takes into account the number and size of the groups and is smaller than that for LASSO.

\subsection{Notation} Let $\|u\| = \langle u,u\rangle^{1/2}$ be the norm induced by the inner product $\langle\cdot,\cdot\rangle$ in the Euclidean space $\Omega$. Given a subspace $\M\subseteq\Omega$ and a point $u\in\Omega$, $u_\M = \Pi_\M(u) = \argmin_{v\in\M}\|u-v\|$ is the projection of $u$ onto $\M$. Denote the orthogonal complement of $\M$ on $\R^p$ as $\M^\perp :=\{u\in\Omega|\langle u,v\rangle = 0, \forall v\in\M\}$. Given some norm $\Rnorm$ on $\Omega$, its associated dual norm is $\Rnorm^*(u) = \sup_{\{v:\Rnorm(v)\le 1\}}\langle u,v\rangle$. For two sets $A,B\subset \Omega$, we define the operation $A+B:=\{u+v|(u,v)\in A\times B\}$. Finally, $\mathfrak{M}$ collects model subspaces.

\section{Problem Formulation}
Let $\Z^n = \{Z_i\}_{i=1}^n$ be a sequence of observations taking values on an Euclidean space $\mathcal{Z}$ and drawn from some distribution $\mathbb{P}$. We are interested in learning a $p$-dimensional parameter vector $\theta^*\in\Omega$, where, typically, $\Omega = \R^p$. Given a convex loss function $\L: \Omega \times \mathcal{Z}^{n} \to \R$, we aim to estimate the unique minimizer of the population risk
\[
    \theta^* = \argmin_{\theta \in \Omega} \E[\L(\theta;\Z^n)].
\]
The dependence of $\L$ on $\Z^n$ is omitted whenever it is clear from the context.

In high-dimensional statistics we often impose restrictions on the parameter space. In other words, the parameter of interest $\theta^*$ does not reside in a $p$-dimensional space, but rather in a smaller subspace of $\Omega$ denoted $\M_*$, meaning that $\theta^* = \theta^*_{\M_*}$. This subspace incorporates the low dimensional constraints imposed to the problem. If $\M_*$ was known beforehand, we could estimate $\widehat{\theta}(\M_*) \in \argmin_{\theta\in\M^*}\L(\theta)$ and our problem would be solved. Unfortunately, $\M_*$ has to be estimated from a very large collection of candidate subspaces.

In order to construct the GIC penalty that incorporate sparsity constraints we use the notion of \textit{decomposable norms} and \textit{subspace compatibility constant} from \cite{sNpRmWbY2012}. 

\begin{definition}[Decomposability]
    Given two subspaces $\M, \Mbar$ in $\Omega$, with  $\M \subseteq \Mbar$, a norm $\Rnorm$ is decomposable with respect to $(\M,\Mbar^\perp)$ if  
    \[\Rnorm(\theta + \gamma) = \Rnorm(\theta) + \Rnorm(\gamma) \text{, for all } (\theta,\gamma) \in (\M,\Mbar^\perp).\]
\end{definition}
For sake of simplicity, we denote $\Mbar$ as the smallest subspace containing $\M$ such that decomposability holds, but results are still valid even if $\Mbar$ is not the smallest possible.

The norm $\Rnorm$ is decomposable with respect to a pair of subspaces $(\M,\Mbar^\perp)$ if it satisfies the decomposability property. \citet{sNpRmWbY2012} and \citet[Chapters 9-12]{wainwright2019High} show that many commonly used norms, such as the $L_1$ (LASSO), the group LASSO ($L_2/L_2$), the nuclear norm for matrices, and overlapping group norms are decomposable with respect to $(\M,\Mbar^\perp)$ selected in a way to take into consideration distinct types of low dimensional restrictions.

The subspace compatibility constant, or subspace Lipschitz constant, ties up the error norm, the decomposable norm, and the restricted subspace. It can be interpreted as a scaling constant necessary to fit every unitary vector in $\M$ to a scaled unit ball in the $\Rnorm$-norm.

\begin{definition}[Subspace Compatibility Constant]
Given a subspace $\M\subseteq\Omega$,
\begin{equation*}\label{eq:SubConst}
    \Psi(\M) := \sup _{u \in \M\setminus\{0\}} \frac{\Rnorm(u)}{\|u\|}
\end{equation*}
is the compatibility constant between the norm $\Rnorm$ and error norm, restricted to $\M$.
\end{definition}

The Subspace Compatibility Constant is connected to fundamental \textit{units} of the model subspace $\M$ and corresponds to the scaling factor required to translate the size of any vector in $\M$ from $\Phi$-norm to the $\|\cdot\|$-norm. Larger, more complex models subspaces require larger constants. We use the square of this quantity as a measure of model complexity, which reduces to known cases in many settings, such as the LASSO \citep{sNpRmWbY2012}.

\subsection{Generalized Information Criteria}

Recall that we aim to estimate $\M_*$, the low-dimensional subspace in which the population parameter $\theta^*$ takes values. Define the set $\Mcoll$ as the collection of all (paired) model subspaces $(\M,\Mbar^\perp)$ with respect to which $\Rnorm$ is decomposable. We emphasize that we take $\M\subseteq\Mbar$ to be the smallest possible, hence selecting $\M$ is equivalent to selecting the pair $(\M,\Mbar^\perp)$. 

The size of $\Mcoll$ can be very large, prompting us to restrict the search space to a smaller collection \cite{chenchen2012, kimetal2012, kimjeon2016}. Suppose the subspace compatibility constant $\Psi(\M_*)\leq\psi_n$, for some $\psi_n>0$, and denote the restricted collection of model subspaces $\mathfrak{M}_{\psi_n^2} = \{\M\in \Mcoll \mid \Psi^2(\M) \leq \psi_n^2\}$. This set collects all models that are not too large. It is equivalent to restricting the number of active regressors in a linear regression problem.

For all model subspaces $\M\in \McollRt$, let
\begin{equation}\label{eq:theta_m}
    \hat{\theta}(\M) \in \argmin_{\theta \in \M}\L(\theta),
\end{equation}
denote the estimated parameter restricted to the model subspace $\M$. Given a non-increasing sequence of positive numbers $a_n$ (that may depend on the dimension of $\M$) we select the model subspace $\widehat{\M}_{a_n}$ minimizing $GIC_{a_n}(\M)$ on $\mathfrak{M}$. More precisely
\begin{equation} \label{eq:GIClike}
    \widehat{\M}_{a_n} = \argmin_{\M \in \McollRt} \underbrace{\L(\hat{\theta}(\M)) + a_n\Psi^2(\M)}_{=:GIC_{a_n}(\M)}, 
\end{equation}
where $\Psi^2(\M)$ is the compatibility constant between $\Rnorm$ and the error norm evaluated at $\M$. We denote $GIC_{a_n}(\M) = \L(\hat{\theta}(\M)) + a_n\Psi^2(\M)$. 

We work with a specific class of \textit{Restricted Strongly Convex (RSC)} functions, treated in \citet{raskutti2010restricted}, in the case of vectors, in \citet{raskutti2011minimax} and \citet{negahban2011estimation} in the case of matrices, and, more generally in \cite[Def. 9.15]{wainwright2019High}. This class is similar to the \emph{quadratic supported risk} class employed in \citet{kimjeon2016} and is a particular case of the RSC class defined in \citet{sNpRmWbY2012}.
\begin{definition}[Restricted Strong Convexity]
For a given norm $\|\cdot\|$ and regularizer $\Rnorm$, the loss function satisfies a Restricted Strong Convexity property with radius $\eta_n$, curvature $\kappa$ and tolerance $\tau_n^2$ if
\begin{equation}\label{eq: RSC}
        \L(\Delta + \theta^*) - \L(\theta^*) -  \langle \nabla\L(\theta^*), \Delta \rangle \geq \kappa\|\Delta\|^2 - \tau_n^2\Rnorm^2(\Delta),
\end{equation}
for all $\Delta \in \Omega: \|\Delta\| \leq \eta_n$, where $\nabla\L(\theta^*)$ is a subgradient of $\L$ at $\theta^*$, and $\Omega$ is a $p$-dimensional Euclidean space.
\end{definition}

In high dimensions it is nearly impossible for strong convexity to hold for all error vectors $\Delta$. The RSC property imposes a milder restriction on the first-order Taylor expansion error's form, by adding the tolerance term $\tau_n^2\Rnorm(\Delta)^2$. %\textcolor{red}{Combined with additional hypotheses on the true model subspace complexity, RSC allows us to require strong convexity only for $\mathfrak{M}_{\psi_n^2}$.} 

\subsection{Regularized m-estimator}\label{m-estimation}
In many settings it is not computationally feasible to enumerate all candidate models in $\Mcoll$. Instead, we can use a regularized $m$-estimation with the sparsity inducing norm:
\begin{equation}\label{eq:regmestimator}
    \widehat{\theta}(\lambda)\in\argmin_{\theta\in\Omega}\left\{\L(\theta)+\lambda\Rnorm({\theta})\right\},
\end{equation}
were $\lambda>0$ is the \emph{regularization parameter} and $\Rnorm:\Omega\rightarrow\R_+$ is a norm. Popular regularization methods, such as the LASSO, group LASSO, low rank regularization, and others fall within this framework.

\citet{sNpRmWbY2012} study this class os problems, where the norm $\Rnorm$ is decomposable with respect to a pair of subspaces $(\M,\Mbar^\perp)$, and the convex loss function $\L$ is in the RSC class.\footnote{\citet{sNpRmWbY2012} work with a more flexible specification of RSC, but point out that ``...for many loss functions it is possible to prove that with high probability...'' our RSC specification holds with $\eta_n = 1$.} We adapt their main result to our setting, where we impose strong sparsity and use a distinct RSC.
\begin{theorem*}{(\citet[Theorem 1]{sNpRmWbY2012})}
    Let $\L$ denote a convex loss satisfying the \emph{restricted strong convexity} assumption, and $\Rnorm$ be decomposable with respect to the pair of subspaces $(\M,\Mbar^\perp)$ with $\M\subseteq\Mbar$. Suppose $\lambda\ge2\Rnorm^*(\nabla\L(\theta^*))$ and the sample size is large enough so $16\tau_n^2\Psi^2(\Mbar)\le \kappa/4$. Then any optimal solution $\widehat{\theta}(\lambda)$ to the convex optimization program \eqref{eq:regmestimator} satisfies the bound
    \[
        \|\widehat{\theta}(\lambda) - \theta^*\|^2 \le 16 \frac{\lambda}{\kappa^2}{\Psi^2(\Mbar)}.
    \]
\end{theorem*}

Instead of enumerating all possible models, we construct a sequence of model subspaces indexed by the regularization parameter $\lambda$. For each $\lambda > 0$ define
\begin{equation}\label{eq:modlambda}
    \M_\lambda = \argmax_{\M \in \mathfrak{M}} \left\{\|\widehat{\theta}(\lambda)_\M\| \Big\vert \min_{\mathcal{S}\in\Mcoll:~\mathcal{S}\subseteq\M}\|\widehat{\theta}(\lambda)_{\mathcal{S}} \| \ge \xi_n\right\}, 
\end{equation}
for $\xi_n$ chosen adequately. This collection of model subspaces arises naturally in sparse estimation. Suppose the sequence $\Mcoll_\lambda = \{\M_\lambda | \lambda>0\}$ contains the true model $\M_*$ .
%with probability converging to one, i.e., it is \textit{pathconsistent}. 
We use the GIC to estimate 
\[
    \widehat{\M} = \argmin_{\M\in\mathfrak{M}_\lambda} GIC_{a_n}(\M),
\]
which is equivalent to selecting $\widehat\lambda$ that satisfy
\[
\widehat{\lambda} = \argmin_{\lambda>0}GIC_{a_n}(\M_\lambda)
\]
and setting $\widehat{\M} = \M_{\widehat\lambda}$.

\section{Main Results}
\subsection{Assumptions}
We introduce the main assumptions used in this paper, further illustrated in Section \ref{examples}.

\begin{assumption}[A1]
   There exists a collection of model subspaces $\Mcoll$ and a norm $\Rnorm$ such that: (1) $\Mcoll$ contains the true model subspace $\M_*$; (2) the norm $\Rnorm$ is decomposable with respect to $(\M,\Mbar^\perp)$ for all $\M\in\Mcoll$;\footnote{Recall that $\Mbar$ is the smallest subspace containing $\M$ such that the decomposability condition is satisfied} (3) the sub-collection $\McollRt$ contains the true model subspace, i.e., $\Psi^2(\M_*)\le \psi_n^2 < \infty$.
\end{assumption}
This condition requires that the sparsity inducing norm $\Rnorm$ can adequately identify the true model subspace from a collection containing it. Furthermore, the true model subspace cannot be too large. 

\begin{assumption}[A2]
    The convex loss function $\L(\theta)$ satisfies the Restricted Strong Convexity property with radius $\eta_n$, curvature $0<\kappa \le 2/5$ and tolerance $\tau_n^2$.
\end{assumption}

\begin{assumption}[A3]
    The decreasing, positive sequence $a_n$ satisfies \[\sqrt{a_n}\ge \frac{2}{\kappa}\Rnorm^*(\nabla\L(\theta^*))\] where $\Rnorm^*(v) = \sup_{\Rnorm(u)\le 1}\langle u, v\rangle$ is the associate dual norm of $\Rnorm$.
\end{assumption}

These condition are traditionally satisfied in a ``\textit{high probability set}''. 
\citet{wainwright2019High} and \cite{sNpRmWbY2012} illustrate this condition in a series of examples, with distinct losses and penalty functions.
%a large number of losses and penalties are shown to satisfy these conditions, and authors derive the respective probability bounds. 
Taking the sample size to infinity, yields a bound with probability converging to one. Also, under stronger conditions, this assumption is satisfied with probability one.

\begin{assumption}[A4]
    $\inf_{\M \in \Mcoll: \M \subseteq \M_*}\|\theta^*_\M\| > \frac{2}{\kappa}\sqrt{a_n}\Psi(\M_*)$
    %$\inf_{\M \in \Mcoll: \M \subseteq \M_* , \|\theta^*_\M\| \neq \{0\}}\|\theta^*_\M\| > \frac{2}{\kappa}\sqrt{a_n}\Psi(\M_*)$
\end{assumption}

Assumption (A4) is a generalization of the well known ``\textit{beta-min condition}'' in the sparse regularization literature. It states that we require a separation between the zero and non-zero parameters in the sparse model, that naturally decreases as the sample size increases. This condition is required for model selection consistency and sign-consistency of the LASSO and other sparse regularized $m$-estimators \cite[Ch. 7]{buhlmann2011statistics}.

\subsection{Model Selection with GIC}

This section provides a deterministic statement about selecting the true model by minimizing the $GIC_{a_n}(\M)$ statistic over a collection $\McollRt$. In practice, however, Assumptions (A2) and (A3) depend on stochastic quantities and hold only probabilistically, i.e., are only true in a \emph{in a set of high probability}. Hence, if assumptions are satisfied in probability or with probability one, one has an asymptotic statement about model selection consistency.

\begin{theorem}[Minimization of GIC]
    Suppose that for $n$ sufficiently large, Assumptions (A1) -- (A4) hold, ${8\tau_n^2\psi_n^2 \le \kappa}$ and $\kappa\eta_n > 4\sqrt{a_n\psi_n^2}$. 
    Then, $GIC_{a_n}(\M^*) < GIC_{a_n}(\M)$ for all $\M\in\Mcoll_{\psi_n^2}\setminus\{\M_*\}$.
    \label{thm:modselprob}
\end{theorem}

This first result shows that, under a set of assumptions, the solution $\widehat{\M}_{a_n}$ to the $GIC_{a_n}$ problem is the true model. Nevertheless, these conditions involve random quantities and must be shown to be satisfied ´´in a set of high probability'', yielding a finite sample probability bound on selecting the true model. If these conditions hold asymptotically, in probability of with probability one, we have consistency.

% Recall
%    \[\widehat{\M} = \argmin_{\M\in\Mcoll_{\psi_n^2}}GIC_{a_n}(\M).\]

In the next definition we use a small abuse of notation. Let $\M_*$ and $\{\M_n\}$ denote a fixed subspace and a sequence of subspaces of $\Omega$, respectively. Suppose $\M_n:=\M(\Z^n)$, we say $\M_n\rightarrow \M_*$ in probability if $\mathbb{P}(\M_n\Delta\M_*)\rightarrow 0$ and  $\M_n\rightarrow \M_*$ with probability one if $\Pr(\lim_n \M_n\Delta\M_*) = 0$.

\begin{definition}[Model Selection Consistency]
    We say that a $GIC_{a_n}$ is model selection consistent if $\widehat{\M}_{a_n}\rightarrow\M_*$ in probability (with probability one) as the sample size increases.  
\end{definition}

 Next Theorem shows that the $GIC_{a_n}$ selects the true model and is model selection consistent. Showing the selected model is the true model is equivalent to showing $GIC_{a_n}(\M^*) < GIC_{a_n}(\M)$ for all $\M\in\mathfrak{M}_{\psi_n^2}\setminus\{\M_*\}$. Naturally, if $\widehat{\M}_{a_n}$ minimizes the $GIC_{a_n}(\cdot)$ over $\McollRt$ it has to be equal to $\M_*$.

\begin{theorem}[Model Selection Consistency]
   Suppose Assumptions (A1) -- (A4) hold in probability (with probability one), ${\tau_n^2\psi_n^2 = o(1)}$ and $\sqrt{a_n\psi_n^2}/\eta_n = o(1) $. Then, $\widehat{\M}_{a_n} \rightarrow \M_*$ in probability (with probability one) as $n\rightarrow\infty$.
   \label{thm:MSC}
\end{theorem}
\begin{proof}
    Result follows directly from Theorem \ref{thm:modselprob}.
\end{proof}

\subsection{Constructing sequence of model subspaces}
In large dimensions it is usually computationally infeasible to enumerate and estimate all possible model subspaces satisfying $\Psi^2(\M)\le\psi_n^2$ and estimate their GICs. Penalized $m$-estimators, discussed in Section \ref{m-estimation}, are a convenient way of constructing a collection of model subspaces indexed by some regularization parameter $\lambda$. We show that under appropriate choice of $\lambda$, we can recover the true model subspace.

Recall we estimate $\widehat{\theta}_\lambda$ solving the convex minimization program \eqref{eq:regmestimator}:
\[
    \widehat{\theta}(\lambda)\in\argmin_{\theta\in\Omega}\left\{\L(\theta)+\lambda\Rnorm({\theta})\right\},~~\lambda>0.
\]
Let $\{\widehat{\theta}(\lambda)\}_{\lambda>0}$ denote a sequence of parameter estimates indexed by $\lambda$. For each $\lambda$, we obtain the model subspace
\[
    \M_\lambda = \argmax_{\M \in \mathfrak{M}} \left\{\|\widehat{\theta}(\lambda)_\M\| \Big\vert \min_{\mathcal{S}\in\Mcoll:~\mathcal{S}\subseteq\M}\|\widehat{\theta}(\lambda)_{\mathcal{S}} \| \ge \xi_n\right\},
\]
for some $\xi_n$ satisfying regularity conditions, and collect them in $\Mcoll_\lambda = \{\M_\lambda|\lambda>0\}$. The model subspace $\M_\lambda$ is interpreted as \emph{the largest model subspace $\M$ such that the error norm of the projection of $\widehat{\theta}_\lambda$ onto every model subspace $\mathcal{S}\subseteq\M$ is larger than a threshold $\xi_n$}. A simpler truncation was used in \cite{kimjeon2016} for the particular cases of LASSO and SCAD. Here we adopt a more general definition adequate to the structured regularization problem.  

Next theorem is a deterministic result showing that $\M_\lambda = \M^*$ for some $\lambda$, under regularity conditions. As in the model selection case, these conditions involve random quantities and can only be satisfied with certain probability or asymptotically. We need a stronger \emph{beta-min} condition involving the regularization parameter $\lambda$ and the complement of the paired model subspace $\bar\M_*$. In practice, this bound is larger than (A4), meaning that (A4') implies (A4) if $\lambda>\kappa\sqrt{a_n}$.
\begin{assumption}[A4']
    For some $c>\frac{3}{2\sqrt{2}}(3+\sqrt{2})$,
    \[
        \inf_{\M \in \Mcoll: \M \subseteq \M_*}\|\theta^*_\M\| > \frac{c}{\kappa}\lambda\Psi(\bar{\M}_*).
    \]
\end{assumption}

\begin{theorem}
    Suppose that for sufficiently large $n$, Assumptions (A1)--(A3) and Assumption (A4') hold, $\lambda\ge \kappa\sqrt{a_n}$, $16\tau_n^2\Psi^2(\bar\M_*)\le \kappa/2$, and $\kappa\eta_n>12\lambda\Psi(\bar\M_*)$. Then $\M_\lambda = \M^*$, provided $\xi_n = \frac{c'}{\kappa}\lambda\Psi(\bar{\M}_*)$, for some $c'>0$.
    \label{thm:pathselection}
\end{theorem}

\begin{definition}[Pathconsistency]
    We say a sequence of model subspaces $\{\M_\lambda|\lambda\in\Lambda\}$ is \textbf{pathconsistent} if $\M_{\lambda^*} = \M_*$ for some $\lambda^*\in\Lambda$, in probability or with probability one, as the number of observations increase.
\end{definition}

We show that the sequence $\Mcoll_\lambda$ is pathconsistent. %We prove that under conditions, there exists $\lambda^*$ such that $\M_{\lambda^*} = \M_*$. Also, if those conditions hold in probability, or with probability one, we have pathconsistency.

\begin{theorem}[Pathconsistency]
    Suppose that Assumptions (A1)--(A3) and Assumption (A4') hold in probability (with probability one), $a_n\lambda^2=o(1)$, $\tau_n^2\Psi^2(\bar\M_*) = o(1)$, and $\lambda\Psi(\bar\M_*)\eta_n^{-1} = o(1)$. Then $\M_\lambda \rightarrow \M^*$, in probability (with probability one) as $n\rightarrow\infty$, provided $\xi_n = \frac{c'}{\kappa}\lambda\Psi(\bar{\M}_*)$, for some $c'>0$.
    \label{thm:pathconsistency}
\end{theorem}
\begin{proof}
    Result follows directly from Theorem \ref{thm:pathselection}.
\end{proof}

\section{Examples}\label{examples}
The most natural way of selecting the sparsity inducing norm $\Rnorm$ is to construct a penalized $m$-estimation problem that aims to recover the desired sparsity pattern. We illustrate the method in two examples. These examples are adapted from the literature in high dimensional statistics, in particular chapters 9 and 10 of \citet{wainwright2019High}.
\subsection{Generalized Linear Models with Group Sparsity}

This example is adapted from \citet[section 9.6]{wainwright2019High}. Let $\Z^n = \{z_i' = (x_i', y_i): i=1,...,n\}$ denote a sequence of random variables where, conditionally on $x_i$, the dependent variables $y_i$ are drawn independently from: 
\[
f(y|x) \propto \exp\left\{\frac{y\langle x,\theta^*\rangle - b(\langle x,\theta^*\rangle)}{c(\sigma)}\right\},    
\]
where $g$ has a bounded second derivative $\|b''\|_\infty \le B^2$. The covariates $x_i\in\R^p$ are independent and identically distributed zero-mean sub-Gaussian random variables. Furthermore, let $\X_g\in\R^{n\times|g|}$ denote the sub-matrix of the covariates indexed by $g$. The covariates satisfy the group normalization condition $\max_{g\in\mathcal{G}}\vertiii{\X_g}_F/\sqrt{n}\le C$, where $\vertiii{\cdot}_F$ is the Frobenius norm.

The group lasso penalty, or $l_1/l_2$-vector norm, is a natural way of modelling group sparsity. Let $\mathcal{G} = \{g_1, ..., g_G\}$ denote a partition of $\{1,...,p\}$ where groups $g_i$ have size at most $m$ and write $\theta_{g_i} = \{\theta_j:j\in g_i\}$. Let $u\in\Omega:=\R^p$, the $l_1/l_2$-norm is
\[
\Rnorm(u) =\sum_{g\in\mathcal{G}}\|\theta_g\|_2\mbox{ with conjugate norm } \Rnorm^*(u) = \max_{g\in\mathcal{G}}\|\theta_g\|_2.
\]
For any $S\subseteq\mathcal{G}$, the space of parameters $\R^{p}$ can be decomposed into
\[
\M:=\M(S) = \{\theta\in\R^p|\theta_{g} = 0\mbox{ for all } g\notin S\},
\]
and $\Mbar^\perp = \M^\perp(S) :=\M(\mathcal{G}\setminus S)$. It follows that $\Rnorm$ is decomposable with respect to $(\M(S),\M^{\perp}(S))$ for any $S\subseteq\mathcal{G}$. 

The subspace compatibility constant is 
\[
    \Psi(\M(S)) = \sup_{\theta\in \M(S)\setminus\{0\}}\frac{\sum_{g\in\mathcal{G}}\|\theta_g\|_2}{\|\theta\|_2} = \sqrt{|S|},
\]
where $|S|$ is the cardinality of $S$.

Under assumptions above, it follows from Theorem 9.16 and Example 9.17 in \citet{wainwright2019High} that RSC condition is satisfied with radius $\eta_n=1$, curvature $\kappa<\infty$ and tolerance $\tau_n^2 = c_1 \left(\frac{m}{n} + \frac{\log G}{n}\right)^2$. Precisely,
\[
    \mathcal{L}(\theta^* + \Delta) - \mathcal{L}(\theta^*) - \langle \nabla\mathcal{L}(\theta),\Delta\rangle \ge \kappa \|\Delta\|^2_2 - c_1 \left(\frac{m}{n} + \frac{\log G}{n}\right)^2\Phi^2(\Delta)
\]
for all $\|\Delta\|\le 1$, with probability at least $1-c_2e^{c_3n}$. The constants $c_1,c_2,c_3,\kappa$ depend on the GLM, the population vector $\theta^*$, $\alpha$ and $\beta$, but are independent of $n$. As $n\rightarrow\infty$ this bound holds with probability converging to one.

As in the proof of \citet[Corolalary 9.28]{wainwright2019High}, it follows from the union bound and standard sub-Gaussian tail bounds that
\[
\Pr\left(2\max_{g\in\mathcal{G}}\|(\nabla\L(\theta^*))_g\|_2 \le \kappa \sqrt{a_n}\right)\ge 1-2e^{-n\delta^2},
\]
where 
\[
    a_n\ge \frac{32B^2C^2}{\kappa^2}\left(\frac{m}{n}\log 5 + \frac{\log G}{n} + \delta^2\right) = O \left(\frac{m}{n} + \frac{\log G}{n}\right),
\]
for adequate choice of $\delta$. Note that if $\log G\rightarrow\infty$ as $n\rightarrow\infty$, the choice $\delta^2 = (\log 5-1)\log(G)/n$ ensures $1-2e^{-n\delta^2}\rightarrow 1$, meaning that (A3) holds with probability converging to one. On the other hand, if $G$ is fixed, we might choose $\delta^2 = O(\log(n)/n)$ which yields $a_n\propto (m+\log G)\log(n)/n$.

Finally, let $\M_* := \M(S_*)$ denote the true model subspace, with support $S_*\subseteq\mathcal{G}$. Assumption (A4) is equivalent to 
\[
\min_{g\in S_*} \|\theta^*_g\|_2 \ge \frac{2}{\kappa}\sqrt{|S_*|a_n}.
\]

In order to apply Theorem \ref{thm:modselprob}, we also require that 
\[
\psi_n^2\le \frac{\kappa}{8c_1}\frac{n}{m + log G} \wedge \frac{\kappa}{4 a_n}.    
\]
This condition is satisfied if $a_n\psi_n^2\rightarrow 0$. In this case, Theorem \ref{thm:MSC} ensures the GIC is model selection consistent. The method is also pathconsistent if we take $\lambda = \kappa\sqrt{a_n}$, $\xi_n \propto \sqrt{a_n|S_*|}$ and construct 
\[
    \M_\lambda:=\argmax_{S\subseteq\mathcal{G}}\left\{\|\hat\theta(\lambda)_{S}\|_2 \Big\vert  \,\forall\,g\in S,~\|\hat\theta(\lambda)_{g}\|_2 > \xi_n\right\}.
\]

\subsection{Low rank matrix regression}

In this example we estimate a low rank matrix regression parameter. It follows developments in \citet[Chapter 10]{wainwright2019High}.

Let $\Z^n = \{(X_i, y_i)\}_{i=1}^n$ be a random sample, where $X_i \in \R^{p_1 \times p_2}$ is a matrix of covariates and $y_i \in \R$ is a response variable. We assume the simplest model, in which every observation pair $(X_i,y_i)$ are linked via the equation:
    \begin{equation*}
        y_i = \langle X_i, \Theta^* \rangle + w_i,
    \end{equation*}
where $\langle A,B\rangle = \trace(A'B)$ is the trace inner product, $w_i$ is a noise variable drawn independently from a zero-mean Normal distribution with variance $\sigma^2$, and $\Theta^* \in \R^{p_1 \times p_2}$ is a matrix of coefficients with rank $r^*<<\min(p_1,p_2)$.

We represent this model more compactly. Define the \textit{observation operator} $\Xfrac: \R^{p_1 \times p_2} \to \R^n$, with elements $[\Xfrac(\Theta^*)]_i := \langle X_i, \Theta^* \rangle$. We write
\begin{equation}\label{eq:matrix_reg_compact}
    y = \Xfrac(\Theta^*) + w,
\end{equation}
where $y \in \R^n$ and $w \in \R^n$ are the vectors formed by stacking the response and noise variables, respectively. We are interested in estimating the matrix of coefficients $\Theta^*$ using least squares, i.e.:
\[   
    \L(\Theta) := \frac{1}{2n}\|y - \Xfrac(\Theta)\|_2^2,
\]
where $\|\cdot\|_2$ is the Euclidean norm in $\R^n$.

Take $\Theta^* \in \R^{p_1\times p_2}$ and assume the matrix $\Theta^*$ has a rank $r_*$ lower than $\min(p_1,p_2)$. For any $\Theta \in \R^{p_1\times p_2}$, denote by $\Row(\Theta) \subseteq \R^{p_2}$ and $\Col(\Theta) \subseteq \R^{p_2}$ the row and column spaces of $\Theta$, respectively. Let $\U$ and $\V$ be a pair of $r$-dimensional subspaces of $\R^{p_1}$ and $\R^{p_2}$, respectively. We define the subspaces
\[
    \M(\U,\V) =\{\Theta \in \R^{p_1\times p_2} \mid \Row(\Theta) \subseteq \V, \Col(\Theta) \subseteq \U\},
\]
and
\[
    \Mbar^\perp(\U,\V) = \{\Theta \in \R^{p_1\times p_2} \mid \Row(\Theta) \subseteq \V^\perp, \Col(\Theta) \subseteq \U^\perp\}.
\]
Notice that, in this case, $\Mbar^\perp \subseteq \M^\perp$, and the sets are generally not equal c.f. \citet{recht2010guaranteed}. For  a fixed pair $(\U,\V)$, we set $\M := \M(\U,\V)$ and $\Mbar^{\perp} := \Mbar^\perp(\U,\V)$. We denote the collection of model subspaces
\[
    \Mcoll = \{(\M(\U,\V),\Mbar^\perp(\U,\V)) \mid \U \subseteq \R^{p_1}, \V \subseteq \R^{p_1} : \Dim(\U) = \Dim(\V) \leq p_1\wedge p_2\},
\]
with pairs $(\M(\U,\V),\Mbar^\perp(\U,\V))$.

The nuclear norm $\|\cdot\|_\sigma$ is defined as 
\[
    \|\Theta\|_\sigma = \sum_{i=1}^{p_1\wedge p_2}\sigma_i(\Theta)\mbox{ with conjugate norm } \vertiii{\Theta}_2 = \max_{i\le p_1\wedge p_2}\sigma_i,
\] 
where $\sigma_i$ is the $i$-th largest singular value of $\Theta$. The conjugate norm is also known as $l_2$-operator norm. For any pair, $(A,B) \in (\M,\Mbar^\perp)$, $A'B = 0$ and $AB' = 0$. \citet{recht2010guaranteed} shows that, under these conditions, the Nuclear norm is additive: $\|A+B\|_\sigma = \|A\|_\sigma + \|B|_\sigma$. Therefore, the nuclear norm is decomposable with respect to any pair of subspaces in $\Mcoll$.

Suppose $\Theta \in \M(\U,\V)$ has $\Dim(\Col(\Theta)) = \Dim(\Row(\Theta)) \leq r$, it follows from the Cauchy-Schwartz inequality
\begin{align*}
    \|\Theta\|_\sigma 
    &= \sum_{k=1}^{p}\sigma_k(\Theta)\\
    &\leq \left(\sum_{k \in \{1,\dots,p\} : \sigma_k(\Theta) \neq 0} 1^2\right)^{\frac{1}{2}}
        \left(\sum_{k=1}^p \sigma_k(\Theta)^2\right)^{\frac{1}{2}} \\
    &\leq \sqrt{\Rank(\Theta)}\vertiii{\Theta}_F\\
    &\leq \sqrt{r}\vertiii{\Theta}_F,
\end{align*}
 with equality holding when $\Theta = USV'$, where $S = \diag(1,\dots,1,0,\dots,0)$ with $r$ entries one and $p-r$ entries zero, and $U \in \R^{p_1 \times p}, V \in \R^{p_2 \times p}$ are orthonormal matrices, with the first $r$ columns of both matrices spanning $\U$ and $\V$, respectively. Therefore we conclude that,
     \[\Psi(\M(U,V)) = \sqrt{r}.\]

\citet{wainwright2019High}[Chapter 10] shows that,for some constants $\kappa/2<1<c_0$, the proposed loss function satisfies
\[
    \L(\Delta + \Theta^*) - \L(\Theta^*) -\langle \nabla\L(\Theta^*) , \Delta \rangle \geq \frac{\kappa}{2}\vertiii{\Delta}_F^2 - c_0\frac{p_1+p_2}{n}\vertiii{\Delta}_\sigma^2,
\]
for all $\Delta \in \R^{p_1 \times p_2}$. In other words, the squared loss satisfy the RSC with curvature $\kappa/2$, tolerance $\tau_n = c_0 (p_1+p_2)/n$ and radius $\eta_n\in \R$. This bound holds with probability at least $1-({1-e^{\frac{n}{32}}})^{-1}$ in the case where the matrix regressors $X_i$ are independent and identically distributed with entries drawn independently from a standard Gaussian distribution. This artificial example can be directly generalized to encompass entries with Gaussian dependence. 

Following same arguments in \citet{wainwright2019High}[Corollary 10.10], under the assumptions above,
\begin{equation*}
    \Pr\left(2\vertiii{\nabla\L(\Theta^*)}_2 \leq \kappa\sqrt{a_n}\right) \geq 1- e^{-\frac{n}{8}}-2^{-n\delta^2},
\end{equation*}
where
\[
a_n \ge \frac{32\sigma^2}{\kappa^2}\left(\frac{p_1+p_2}{n}log 9 + \delta^2\right) = O\left(\frac{p_1+p_2}{n}\right),
\]
for adequate choice of $\delta^2$. Note that if $p_1+p_2\rightarrow\infty$ as $n\rightarrow\infty$ the choice $\delta^2 = (\log 9 -1)(p_1+p_2)/n$ ensures $1- e^{-n/8}-2^{-n\delta^2}\rightarrow 1$, meaning that Assumption (A3) is satisfied with probability converging to one. On the other hand, if $p_1+p_2$ is fixed, we might choose $\delta^2 \propto (p_1+p_2)\log(n)/n$.

Finally, let $\sigma_1^*\ge\cdots\ge\sigma_{r_*}^*>\sigma_{r_*+1}^*= 0 =\cdots = 0$ denote the population singular values of $\Theta^*$. Assumption (A4) is satisfied if 
\[
    \sigma_{r_*}^*\ge \frac{2}{\kappa}\sqrt{a_n r_*}.
\]

In order to apply Theorem \ref{thm:modselprob}, we also require that 
\[
\psi_n^2\le \frac{\kappa}{8c_0}\frac{n}{p_1+ p_2} \wedge \frac{\kappa}{4 a_n}.    
\]
This condition is satisfied if $a_n\psi_n^2\rightarrow 0$. In this case, Theorem \ref{thm:MSC} ensures the GIC is model selection consistent. In this case, Assumption (A4) and (A4') differ. However, the rank of $\Mbar_*$ and $\M_*$ is the same, hence $\Psi(\Mbar_*) = \Psi(\M_*)$. The method is pathconsistent if we take $\lambda = \kappa\sqrt{a_n}$, $\xi_n \propto \sqrt{a_nr_*}$ and construct 
\[
    \M_\lambda:=\argmax_{\M\in\Mcoll}\left\{\vertiii{\hat\Theta(\lambda)_{\M}}_F \Big\vert  \,\forall\,\mathcal{S}\subseteq\M ,~\vertiii{\hat\Theta(\lambda)_\mathcal{S}}_F > \xi_n\right\}.
\]

\section{Proofs}
In this section we prove Theorem \ref{thm:modselprob} and Theorem \ref{thm:pathselection}.

\subsection{Proof of Theorem \ref{thm:modselprob}}   
\begin{proof}
    We first show that the solution $\widehat{\theta}(\widehat{\M})$ must be inside a ball around $\theta^*$, i.e. set $\widehat{\Delta} = \widehat{\theta}(\widehat{\M}) - \theta^*$, $\|\widehat{\Delta}\|\le\eta_n$. 
    
    Let $\theta\in\M\in\mathfrak{M}_{\psi_n^2}$ be such that $\|\Delta\|>\eta_n$, with $\Delta = \theta-\theta^*$. Let $\theta_h = \theta^* + h(\theta - \theta^*)$, and $\Delta_h = \theta_h - \theta^* = h\Delta$. Take $h = \tfrac{\eta_n}{\|\theta - \theta^*\|}$, meaning that $\|\Delta_h\| = \eta_n$. Hence, it follows from the RSC condition
    \begin{align*}
        \L(\theta_h) - \L(\theta^*) &\geq \langle\nabla\L(\theta^*),\Delta_h \rangle + \kappa\|\Delta_h\|^2 - \tau_n^2\Rnorm^2(\Delta_h)\\
        &{\geq} -\Rnorm^*(\nabla\L(\theta^*))\Rnorm(\Delta)h + \kappa h^2\|\Delta\|^2 - \tau_n^2h^2\Rnorm^2(\Delta)\\
        &= -\frac{\Rnorm^*(\nabla\L(\theta^*))\Rnorm(\Delta)\eta_n}{\|\Delta\|} + \kappa \eta_n^2 - \tau_n^2\eta_n^2\frac{\Rnorm^2(\Delta)}{\|\Delta\|^2}\\
        &{\geq} -\frac{\sqrt{a_n}\kappa\Psi(\M + \M_*)\eta_n}{2} + \kappa \eta_n^2 - \tau_n^2\eta_n^2\Psi^2(\M + \M_*)\\
        &\geq -\sqrt{a_n\psi_n^2}\kappa\eta_n + \kappa \eta_n^2 - 4\tau_n^2\eta_n^2\psi_n^2\\
        &\geq -\sqrt{a_n\psi_n^2}\kappa\eta_n + \kappa \eta_n^2\left\{1 - 4\frac{\tau_n^2}{\kappa}\psi_n^2 \right\}\\
        &\geq -\sqrt{a_n\psi_n^2}\kappa\eta_n + \frac{\kappa \eta_n^2}{2}\\
        &\geq \frac{\kappa \eta_n^2}{2}\left\{1 - \frac{\sqrt{4a_n\psi_n^2}}{\eta_n}\right\}\\
        &\geq \frac{\kappa \eta_n^2}{4}.
    \end{align*}
    
    Recall $\L$ is a convex function of $\theta$, so $h(\L(\theta)-\L(\theta^*)) \geq \L(\theta_h)-\L(\theta^*) \geq  \frac{\kappa \eta_n^2}{4}$. Combining this bound with Assumption (A1) and $16\psi_na_n\le \kappa\eta_n^2$,
    \begin{align*}
    \L(\theta) - \L(\theta^*) + a_n (\Psi^2(\M) - \Psi^2(\M_*))
        &\geq \frac{\kappa \eta_n^2}{4}-a_n\Psi^2(\M_*)\\ 
        &\geq \frac{\kappa \eta_n^2}{4}\left\{1 - 4\frac{\psi_na_n}{\kappa \eta_n^2}\right\}\\
        &\geq \frac{\kappa \eta_n^2}{8} > 0.
    \end{align*}
    
    Write
    \begin{align*}
        GIC_{a_n}(\M) &- GIC_{a_n}(\M_*) = \L(\widehat{\theta}_\M)-\L(\widehat\theta_{\M_*}) + a_n(\Psi^2(\M)-\Psi^2(\M_*))\\
        &\ge \L(\widehat{\theta}_\M)-\L(\theta^*) + a_n(\Psi^2(\M)-\Psi^2(\M_*),
    \end{align*}
    given $\L(\widehat\theta_{\M_*}) \le \L(\theta^*)$. 
    
    If $\widehat{\theta}(\widehat{\M})$ is outside the ball around $\theta^*$, then $GIC_{a_n}(\widehat{\M}) - GIC_{a_n}(\M_*) >0$. It contradicts the premise that $\widehat{\M}$ minimizes the GIC. Therefore it cannot true that $\left\|\widehat{\theta}(\widehat{\M}) - \theta^*\right\| > \eta_n$.

    Now, given we know the solution is inside a ball around $\theta^*$, we show that the solution must be in $\M_*$. Let $\M\in\mathfrak{M}_{\psi_n^2}$ be arbitrary, $\theta\in\M$ with $\|\theta-\theta_*\|\le \eta_n$. It follows from Assumption (A2) (Restricted Strong Convexity) with $\Delta = \theta - \theta^*$ 
    \begin{align*}
        \L(\theta)-\L(\theta^*) 
        &\ge \langle\nabla\L(\theta^*),\Rnorm(\Delta)\rangle +\kappa\|\Delta\|^2 -\tau_n^2\Rnorm(\Delta)^2 \\
        &\ge -\Rnorm^*(\nabla\L(\theta^*))\Rnorm(\Delta) +\kappa\|\Delta\|^2 -\tau_n^2\Rnorm(\Delta)^2 \\
        &\ge -\frac{\kappa \sqrt{a_n}}{2}\Rnorm(\Delta) +\kappa \|\Delta\|^2-\tau_n^2\Rnorm(\Delta)^2\\
    \end{align*}
    Split $\M+\M_*$ into four orthogonal subspaces (with respect to $\langle\cdot,\cdot\rangle$) denoted $M_1,...,M_4$. Then, for all $\gamma \in \R^p$, $\|\gamma\|^2 = \sum_{i=1}^4\|\gamma_{M_i}\|$. Recall that $\widehat{\theta}_{\M^\perp} = \bs{0}$ and $\theta^*_{\M_*^\perp} = 0$, setting
    \[\M + \M_* = \underbrace{(\M \cap \M_*)}_{M_1} \oplus \underbrace{(\M^\perp \cap \M_*)}_{M_2} \oplus \underbrace{(\M \cap \M_*^\perp)}_{M_3} \oplus \underbrace{(\M^\perp \cap \M_*^\perp)}_{M_4},\]
    yield 
    $\Delta_{M_1} = \widehat{\theta}_\M - \theta^*_{\M_*}$, $\Delta_{M_2} = \theta^*_{\M_*}$,  $\Delta_{M_3} = \widehat{\theta}_{\M}$ and 
    $ \Delta_{M_4} = 0$. Moreover, 
    \begin{align*}
        \L(\widehat{\theta}_\M) &-\L(\theta^*) + a_n(\Psi^2(\M)-\Psi^2(\M_*))\\
        & \ge   \underbrace{-\frac{\kappa \sqrt{a_n}}{2}\Rnorm(\Delta_{M_1}) +\kappa \|\Delta_{M_1}\|^2-3\tau_n^2\Rnorm(\Delta_{M_1})^2}_{I}\\
        & \underbrace{-\frac{\kappa \sqrt{a_n}}{2}\Rnorm(\Delta_{M_2}) +\kappa\|\Delta_{M_2}\|^2 -3\tau_n^2\Rnorm(\Delta_{M_2})^2 - a_n\Psi^2(\M_*)}_{II}\\
        & \underbrace{-\frac{\kappa \sqrt{a_n}}{2}\Rnorm(\Delta_{M_3}) +\kappa\|\Delta_{M_3}\|^2 -3\tau_n^2\Rnorm(\Delta_{M_3})^2 +a_n\Psi^2(\M)}_{III}
    \end{align*}
    
    We must bound $I - III$ individually. It follows from the definition of the compatibility constant that $\Rnorm(\Delta_{M_1})^2 \le \|\Delta_{M_1}\|^2\Psi^2(\M_*)$ and under assumption $\kappa - 3\tau_n^2\Psi^2(\M_*) \ge \kappa/2$. Therefore
    \begin{align*}
        I &= -\frac{\kappa \sqrt{a_n}}{2}\Rnorm(\Delta_{M_1}) +\kappa \|\Delta_{M_1}\|^2-3\tau_n^2\Rnorm(\Delta_{M_1})^2\\
        &\ge -\frac{\kappa \sqrt{a_n}}{2}\Psi(M_1)\|\Delta_{M_1}\| +\frac{\kappa}{2}\|\Delta_{M_1}\|^2\\
        &\ge - \frac{\kappa a_n}{8}\Psi^2(\M_*\cap\M) \ge - \frac{\kappa a_n}{8}\Psi^2(\M) .
    \end{align*}
    
    Similarly, $\Rnorm(\Delta_{M_2})^2 \le \|\Delta_{M_2}\|^2\Psi^2(\M_*)$, $\kappa - 3\tau_n^2\Psi^2(\M_*) \ge \kappa/2$  and Assumption (A4) yield
    \begin{align*}
        II & = -\frac{\kappa \sqrt{a_n}}{2}\Rnorm(\Delta_{M_2}) +\kappa\|\Delta_{M_2}\|^2 -3\tau_n^2\Rnorm(\Delta_{M_2})^2 - a_n\Psi^2(\M_*)\\
        &\ge -\frac{\kappa \sqrt{a_n}}{2}\Psi(\M_*)\|\theta^*_{\M^*}\| +\frac{\kappa}{2}\|\theta^*_{\M_*}\|^2 - a_n\Psi^2(\M_*)\\
        &= \frac{\kappa}{2}\|\theta^*_{\M_*}\|^2\left[1-\frac{\sqrt{a_n}\Psi^2(\M_*)}{\|\theta^*_{\M_*}\|} - \frac{2a_n\Psi^2(\M_*)}{\kappa\|\theta^*_{\M_*}\|^2}\right]\\
        &\overset{(A4)}{\ge} 2\left(\frac{1-\kappa}{\kappa}\right)a_n\Psi^2(\M_*).
    \end{align*}
    
    Finally, $\Rnorm(\Delta_{M_3})^2 \le \|\Delta_{M_3}\|^2\Psi^2(\M)$ and $\kappa - 3\tau_n^2\Psi^2(\M) \ge \kappa/2
    $    
    \begin{align*}
        III &= -\frac{\kappa \sqrt{a_n}}{2}\Rnorm(\Delta_{M_3}) +\kappa\|\Delta_{M_3}\|^2 -3\tau_n^2\Rnorm(\Delta_{M_3})^2 +a_n\Psi^2(\M)\\
        &\ge -\frac{\kappa \sqrt{a_n}}{2}\Psi(\M)\|\Delta_{M_3}\| + (\kappa  -3\tau_n^2\Psi^2(\M))\|\Delta_{M_3}\|^2 +a_n\Psi^2(\M)\\
        &\ge -\frac{\kappa \sqrt{a_n}}{2}\Psi(\M)\|\Delta_{M_3}\| +\frac{\kappa}{2}\|\Delta_{M_3}\|^2 +a_n\Psi^2(\M)\\
        &\ge \left(1-\frac{\kappa}{8}\right)a_n\Psi^2(\M).
    \end{align*}
    
    Now suppose $\kappa \le 2/3$, then
    \[
    I+II+III \ge \left(\frac{8-\kappa}{16}\right)a_n 2(\Psi^2(\M_*)+\Psi^2(\M)) \ge  \left(\frac{8-\kappa}{16}\right)a_n\Psi^2(\M_*+\M).
    \]

    Note that showing $\widehat{\M}=\M_*$ is equivalent to showing that $GIC_{a_n}(\M_*) < GIC_{a_n}(\M)$ for all $\M\in\mathfrak{M}_{\psi_n^2}\setminus\{\M_*\}$. Hence, the result follows.
    
    %The second statement follows from observing Assumptions (A2) and (A3) typically hold in a set with high probability and, in the limit, may hold in probability or probability converging to one.
\end{proof}

\subsection{Proof of Theorem \ref{thm:pathselection}}
\begin{proof}
    Under Assumption (A3), $\lambda\ge\kappa\sqrt{a_n}$ implies $2\Rnorm^*(\nabla\L(\theta^*)) \le \lambda$. 
    It follows from \cite[Lemma 1]{sNpRmWbY2012} (or \cite[Proposition 9.13]{wainwright2019High}) that $\widehat\Delta_\lambda = \widehat{\theta}(\lambda)-\theta^*$ is inside
    \[\mathbb{C}(\M,\bar{\M}^{\perp};\theta^*) = \left\{\Delta\in\R^p\vert \Rnorm(\Delta_{\bar{\M}^\perp}) \le 3\Rnorm(\Delta_{\bar{\M}}) + 4\Rnorm(\theta^*_{\M^\perp})\right\},  \]
    where $\Rnorm$ is decomposable with respect to $(\M,\bar{\M}^\perp)$. 
    Under the sparsity assumption, $\Rnorm(\Delta)\le 4\Rnorm(\Delta_{\bar{\M}_*})\le\Psi(\bar{\M}_*)\|\Delta_{\bar{\M}_*}\|\le 4\Psi(\bar{\M}_*)\|\Delta\|$, for all $\Delta\in\mathbb{C}(\M_*,\bar{\M}_*^\perp;\theta^*)$.

    First we show that $\|\widehat\Delta_\lambda\|\le\eta_n$ using arguments parallel to those in the proof of Theorem \ref{thm:MSC}. Select $\Delta\in\mathbb{C}(\M_*,\bar{\M}_*^\perp;\theta^*)$ such that $\|\Delta\|>\eta_n$ and define $\Delta_h = \theta^* + h\Delta$ with $h = \eta_n/\|\Delta\|$. It follows from the RSC and reverse triangle inequality
    \begin{align*}
        \L(\theta_h) &-\L(\theta^*) + \lambda(\Rnorm(\theta_h)-\Rnorm(\theta^*))\\
        & \ge -\Rnorm^*(\nabla\L(\theta^*))\Rnorm(\Delta_h) + \kappa\|\Delta_h\|^2 - \tau_n^2\Rnorm^2(\Delta_h) - \lambda\Rnorm(\Delta_h)\\
        & \ge -6\lambda\Psi(\bar{\M}_*)\eta_n + (\kappa - 16\tau_n^2\Psi^2(\bar{\M}_*))\eta_n^2>0, 
    \end{align*}
    provided $\eta_n > 12 \lambda\Psi(\bar{\M}_*)/\kappa$, and  $16\tau_n^2\Psi^2(\bar{\M}_*)<\kappa/2$. Convexity of $\mathcal{F}(\theta) = \L(\theta) -\L(\theta^*) + \lambda(\Rnorm(\theta)-\Rnorm(\theta^*))$ yields a contradiction as $\widehat{\theta}(\lambda)$ is a minimizer. Hence, it must be true that $\|\widehat\Delta_\lambda\|\le\eta_n$

    %Now we must show that the choice $\tau = \frac{5}{4}\sqrt{a_n}\Psi(\bar{\M}_*)$ yields a pathconsistent sequence. 
    
    Recall $\Rnorm(\cdot)$ is decomposable with respect to $(\M_*,\bar{\M}_*^\perp)$ with $\M_*\subseteq\bar{\M}$, and that $\theta^*_{\M_*^\perp} = \bs{0}$. It follows from RSC and \citet[Lemma 9.14]{wainwright2019High}
    \begin{align*}
        \L(\theta) &-\L(\theta^*) + \lambda(\Rnorm(\theta)-\Rnorm(\theta^*))\\
        & \ge\langle\nabla\L(\theta^*),\Delta\rangle + \kappa\|\Delta\|^2 -\tau_n^2\Rnorm^2(\Delta) +\lambda(\Rnorm(\Delta_{\bar{\M}_*^\perp})-\Rnorm(\Delta_{\M_*}))\\
        & \ge  -\Rnorm^*(\nabla\L(\theta^*))\Rnorm(\Delta_{\bar{\M}_*}) + (\kappa-16\tau_n^2\Psi^2(\bar{\M}_*))\|\Delta_{\bar{\M}_*}\|^2 -\lambda\Rnorm(\Delta_{\bar{\M}_*})\\
        &\quad -\Rnorm^*(\nabla\L(\theta^*))\Rnorm(\Delta_{\bar{\M}_*^\perp}) + \kappa\|\Delta_{\bar{\M}_*^\perp}\|^2 +\lambda\Rnorm(\Delta_{\bar{\M}_*^\perp})\\
        & \ge \underbrace{-\frac{3}{2}\lambda\Rnorm(\Delta_{\bar{\M}_*}) + \frac{\kappa}{2}\|\Delta_{\bar{\M}_*}\|^2}_{I} +\underbrace{\frac{\lambda}{2}\Rnorm(\theta_{\bar{\M}_*^\perp}) + \kappa\|\theta_{\bar{\M}_*^\perp}\|^2}_{II}.
    \end{align*}

    We will show that $\max_{F^\perp\subseteq\bar\M^\perp_*}\|\widehat{\theta}(\lambda)_{F^\perp}\|\le \frac{c_1}{\kappa}\lambda\Psi(\bar\M_*)$. Let $F^\perp\subseteq\bar{\M}_*^\perp$ be such that $\|\theta_{F^\perp}\|>\frac{c_1}{\kappa}\lambda\Psi(\bar{\M}_*)$, for $c_1>3/2\sqrt{2}$. A simple lowed bound on the quadratic equation yields $I\ge-\frac{9}{8\kappa}\lambda^2\Psi^2(\bar\M_*)$. Then,
    \begin{align*}
        \L(\theta) &-\L(\theta^*) + \lambda(\Rnorm(\theta)-\Rnorm(\theta^*))\\
        &\ge -\frac{9}{8\kappa}\lambda^2\Psi^2(\bar\M_*) + \frac{\lambda}{2}\Rnorm(\Delta_{\bar{\M}_*^\perp}) + \kappa\|\Delta_{\bar{\M}_*^\perp}\|^2\\
        &\ge -\frac{9}{8\kappa}\lambda^2\Psi^2(\bar\M_*) + \frac{\lambda}{2}\Rnorm(\theta_{F^\perp}) + \kappa\|\theta_{F^\perp}\|^2\\
        &\ge -\frac{9}{8\kappa}\lambda^2\Psi^2(\bar\M_*) + \frac{c_1}{\kappa}\lambda^2\Psi^2(\bar\M_*)\\
        &\ge \frac{\lambda^2\Psi^2(\bar{\M}_*)}{\kappa}\left(c_1^2 - \frac{9}{8}\right) >0,
    \end{align*}
    which cannot be true for the minimizer $\widehat{\theta}(\lambda)$, showing that \begin{equation}\label{eq:ubthetahat}
        \max_{F^\perp\subseteq\bar\M^\perp_*}\|\widehat{\theta}(\lambda)_{F^\perp}\|\le \frac{3}{2\sqrt{2}\kappa}\lambda\Psi(\bar\M_*).
    \end{equation}
    
    Now we show that $\min_{F\in\Mcoll:F\in\M_*}\|\widehat{\theta}(\lambda)_F\| \ge \frac{c'}{\kappa}\lambda\Psi(\bar\M_*)$. Let $F\subseteq\M_*$ satisfy $\|\Delta_F\| > \frac{c_2}{\kappa}\lambda\Psi(\bar{\M}_*)$, for $c_2 \ge \frac{3}{2}(1+\sqrt{2})$. Once again, because $II\ge 0$,
    \begin{align*}
        \L(\theta) &-\L(\theta^*) + \lambda(\Rnorm(\theta)-\Rnorm(\theta^*))\\
        &\ge -\frac{3}{2}\lambda\Rnorm(\Delta_{\bar{\M}_*}) + \frac{\kappa}{2}\|\Delta_{\bar{\M}_*}\|^2\\
        &\ge -\frac{3}{2}\kappa\sqrt{a_n}\left(\Rnorm(\Delta_{\bar{\M}_*\cap F}) +\Rnorm(\Delta_{\bar{\M}_*\cap F^\perp})\right) + \frac{\kappa}{2}\left(\|\Delta_{\bar{\M}_*\cap F}\|^2+ \|\Delta_{\bar{\M}_*\cap F^\perp}\|^2\right)\\
        &\ge -\frac{3}{2}\lambda\Psi(F)\|\Delta_F\| + \frac{\kappa}{2}\|\Delta_F\|^2 - \frac{9}{8} \lambda \Psi^2(\bar{\M}_*\cap F^\perp) >0
    \end{align*}
    under the bound on $\|\Delta_F\|$. It cannot be true for the minimizer $\widehat{\theta}(\lambda)$, hence, $\|\widehat{\Delta}_{\lambda,F} \|\le \frac{c_2}{\kappa} \lambda\Psi(\bar{\M}_*)$. Finally, choose $c_3>c_1+c_2$. The reverse triangle inequality yields $\|\widehat{\theta}(\lambda)_F\|\le \|\theta^*_F\| + \frac{c_2}{\kappa}\lambda\Psi(\bar\M_*)$ and $\|\widehat{\theta}(\lambda)_F\|\ge (c_3-c_2)\frac{\lambda\Psi(\bar\M_*)}{\kappa}$. It follows that there is $c' > c_1$ such that 
    \begin{equation}\label{eq:lbthetahat}
        \min_{F\in\M_*}\|\widehat{\theta}(\lambda)_F\| \ge \frac{c'}{\kappa}\lambda\Psi(\bar\M_*) =: \xi_n,
    \end{equation} 
    as required.

    Finally, we set $\M_\lambda$ the largest subspace $\M\in\Mcoll$ satisfying the bound $\min_{F\in\Mcoll:~\subseteq\M}\|\widehat{\theta}(\lambda)_F\| \ge\xi_n$. 
\end{proof}

\bibliography{mybib}

\begin{thebibliography}{35}
\providecommand{\natexlab}[1]{#1}
\providecommand{\url}[1]{\texttt{#1}}
\expandafter\ifx\csname urlstyle\endcsname\relax
  \providecommand{\doi}[1]{doi: #1}\else
  \providecommand{\doi}{doi: \begingroup \urlstyle{rm}\Url}\fi

\bibitem[{Akaike}(1974)]{AkaikeAIC}
H.~{Akaike}.
\newblock A new look at the statistical model identification.
\newblock \emph{IEEE Transactions on Automatic Control}, 19\penalty0
  (6):\penalty0 716--723, December 1974.
\newblock ISSN 0018-9286.
\newblock \doi{10.1109/TAC.1974.1100705}.

\bibitem[Broman and Speed(2002)]{bromanspeed2002}
K.~W. Broman and T.~P. Speed.
\newblock A model selection approach for the identification of quantitative
  trait loci in experimental crosses.
\newblock \emph{Journal of the Royal Statistical Society: Series B (Statistical
  Methodology)}, 64\penalty0 (4):\penalty0 641--656, 2002.

\bibitem[B{\"u}hlmann and Van De~Geer(2011)]{buhlmann2011statistics}
P.~B{\"u}hlmann and S.~Van De~Geer.
\newblock \emph{Statistics for high-dimensional data: methods, theory and
  applications}.
\newblock Springer Science \& Business Media, 2011.

\bibitem[Casella et~al.(2009)Casella, Gir{\'o}n, Mart{\'\i}nez, Moreno,
  et~al.]{casellaetal2009}
G.~Casella, F.~J. Gir{\'o}n, M.~L. Mart{\'\i}nez, E.~Moreno, et~al.
\newblock Consistency of bayesian procedures for variable selection.
\newblock \emph{The Annals of Statistics}, 37\penalty0 (3):\penalty0
  1207--1228, 2009.

\bibitem[Chen and Chen(2008)]{chenchen2008}
J.~Chen and Z.~Chen.
\newblock {Extended Bayesian Information Criteria for Model Selection with
  Large Model Spaces model selection Extended Bayesian information criteria for
  with large model spaces}.
\newblock \emph{Biometrika}, 95\penalty0 (3):\penalty0 759--771, 2008.

\bibitem[Chen and Chen(2012)]{chenchen2012}
J.~Chen and Z.~Chen.
\newblock {Extended BIC for small-n-large-P sparse GLM}.
\newblock \emph{Statistica Sinica}, 22\penalty0 (2), 2012.
\newblock URL
  \url{http://www3.stat.sinica.edu.tw/statistica/j22n2/J22N26/J22N26.html}.

\bibitem[Chen et~al.(2001)Chen, Donoho, and Saunders]{chenetal2001}
S.~S. Chen, D.~L. Donoho, and M.~A. Saunders.
\newblock Atomic decomposition by basis pursuit.
\newblock \emph{SIAM review}, 43\penalty0 (1):\penalty0 129--159, 2001.

\bibitem[Craven and Wahba(1978)]{cravenwhaba1978}
P.~Craven and G.~Wahba.
\newblock Smoothing noisy data with spline functions.
\newblock \emph{Numerische mathematik}, 31\penalty0 (4):\penalty0 377--403,
  1978.

\bibitem[Fan and Li(2001)]{fan2001variable}
J.~Fan and R.~Li.
\newblock Variable selection via nonconcave penalized likelihood and its oracle
  properties.
\newblock \emph{Journal of the American statistical Association}, 96\penalty0
  (456):\penalty0 1348--1360, 2001.

\bibitem[Foygel and Drton(2010)]{foygeldrton2010}
R.~Foygel and M.~Drton.
\newblock Extended bayesian information criteria for gaussian graphical models.
\newblock In \emph{Advances in neural information processing systems}, pages
  604--612, 2010.

\bibitem[Friedman et~al.(2008)Friedman, Hastie, and
  Tibshirani]{friedmanetal2008}
J.~Friedman, T.~Hastie, and R.~Tibshirani.
\newblock Sparse inverse covariance estimation with the graphical lasso.
\newblock \emph{Biostatistics}, 9\penalty0 (3):\penalty0 432--441, 2008.

\bibitem[Gao and {X-K Song}(2010)]{gaosong2010}
X.~Gao and P.~{X-K Song}.
\newblock {Composite Likelihood Bayesian Information Criteria for Model
  Selection in High-Dimensional Data}.
\newblock \emph{Journal of the American Statistical Association}, 105\penalty0
  (492):\penalty0 1531--1540, 2010.

\bibitem[Kim and Jeon(2016)]{kimjeon2016}
Y.~Kim and J.-J. Jeon.
\newblock Consistent model selection criteria for quadratically supported
  risks.
\newblock \emph{The Annals of Statistics}, 44\penalty0 (6):\penalty0
  2467--2496, 2016.

\bibitem[Kim et~al.(2012)Kim, Kwon, and Choi]{kimetal2012}
Y.~Kim, S.~Kwon, and H.~Choi.
\newblock {Consistent Model Selection Criteria on High Dimensions}.
\newblock \emph{Journal of Machine Learning Research}, 13:\penalty0 1037--1057,
  2012.

\bibitem[Lee et~al.(2014)Lee, Noh, and Park]{leeetal2014}
E.~R. Lee, H.~Noh, and B.~U. Park.
\newblock Model selection via bayesian information criterion for quantile
  regression models.
\newblock \emph{Journal of the American Statistical Association}, 109\penalty0
  (505):\penalty0 216--229, 2014.

\bibitem[Meinshausen et~al.(2006)Meinshausen, B{\"u}hlmann, et~al.]{nMpB2006}
N.~Meinshausen, P.~B{\"u}hlmann, et~al.
\newblock High-dimensional graphs and variable selection with the lasso.
\newblock \emph{The annals of statistics}, 34\penalty0 (3):\penalty0
  1436--1462, 2006.

\bibitem[Negahban et~al.(2011)Negahban, Wainwright,
  et~al.]{negahban2011estimation}
S.~Negahban, M.~J. Wainwright, et~al.
\newblock Estimation of (near) low-rank matrices with noise and
  high-dimensional scaling.
\newblock \emph{The Annals of Statistics}, 39\penalty0 (2):\penalty0
  1069--1097, 2011.

\bibitem[Negahban et~al.(2012)Negahban, Ravikumar, Wainwright, and
  Yu]{sNpRmWbY2012}
S.~N. Negahban, P.~Ravikumar, M.~J. Wainwright, and B.~Yu.
\newblock A unified framework for high-dimensional analysis of $ m $-estimators
  with decomposable regularizers.
\newblock \emph{Statistical Science}, 27\penalty0 (4):\penalty0 538--557, 2012.

\bibitem[Ni and Cai(2018)]{nicai2018}
A.~Ni and J.~Cai.
\newblock Tuning parameter selection in cox proportional hazards model with a
  diverging number of parameters.
\newblock \emph{Scandinavian Journal of Statistics}, 45\penalty0 (3):\penalty0
  557--570, 2018.

\bibitem[Raskutti et~al.(2010)Raskutti, Wainwright, and
  Yu]{raskutti2010restricted}
G.~Raskutti, M.~J. Wainwright, and B.~Yu.
\newblock Restricted eigenvalue properties for correlated gaussian designs.
\newblock \emph{Journal of Machine Learning Research}, 11\penalty0
  (Aug):\penalty0 2241--2259, 2010.

\bibitem[Raskutti et~al.(2011)Raskutti, Wainwright, and
  Yu]{raskutti2011minimax}
G.~Raskutti, M.~J. Wainwright, and B.~Yu.
\newblock Minimax rates of estimation for high-dimensional linear regression
  over lq balls.
\newblock \emph{IEEE transactions on information theory}, 57\penalty0
  (10):\penalty0 6976--6994, 2011.

\bibitem[Recht et~al.(2010)Recht, Fazel, and Parrilo]{recht2010guaranteed}
B.~Recht, M.~Fazel, and P.~A. Parrilo.
\newblock Guaranteed minimum-rank solutions of linear matrix equations via
  nuclear norm minimization.
\newblock \emph{SIAM review}, 52\penalty0 (3):\penalty0 471--501, 2010.

\bibitem[Schwarz(1978)]{schwarz1978}
G.~Schwarz.
\newblock Estimating the dimension of a model.
\newblock \emph{Ann. Statist.}, 6\penalty0 (2):\penalty0 461--464, 03 1978.
\newblock \doi{10.1214/aos/1176344136}.
\newblock URL \url{https://doi.org/10.1214/aos/1176344136}.

\bibitem[She and Tran(2019)]{shetran2019}
Y.~She and H.~Tran.
\newblock On cross-validation for sparse reduced rank regression.
\newblock \emph{Journal of the Royal Statistical Society: Series B (Statistical
  Methodology)}, 81\penalty0 (1):\penalty0 145--161, 2019.

\bibitem[Stone(1974)]{stone1974}
M.~Stone.
\newblock Cross-validatory choice and assessment of statistical predictions.
\newblock \emph{Journal of the Royal Statistical Society: Series B
  (Methodological)}, 36\penalty0 (2):\penalty0 111--133, 1974.

\bibitem[Tibshirani(1996)]{tibshirani1996}
R.~Tibshirani.
\newblock Regression shrinkage and selection via the {LASSO}.
\newblock \emph{Journal of the Royal Statistical Society. Series {B}
  (Methodological)}, 58:\penalty0 267--288, 1996.

\bibitem[Tibshirani et~al.(2015)Tibshirani, Wainwright, and
  Hastie]{tibshirani2015statistical}
R.~Tibshirani, M.~Wainwright, and T.~Hastie.
\newblock \emph{Statistical learning with sparsity: the lasso and
  generalizations}.
\newblock Chapman and Hall/CRC, 2015.

\bibitem[Wainwright(2019)]{wainwright2019High}
M.~J. Wainwright.
\newblock \emph{High-Dimensional Statistics: A Non-Asymptotic Viewpoint}.
\newblock Cambridge series in statistical and probabilistic mathematics.
  Cambridge University Press, 02 2019.
\newblock ISBN 9781108498029.
\newblock \doi{10.1017/9781108627771}.

\bibitem[Wang et~al.(2007)Wang, Li, and Tsai]{wangetal2007}
H.~Wang, G.~Li, and C.~L. Tsai.
\newblock {Regression coefficient and autoregressive order shrinkage and
  selection via the lasso}.
\newblock \emph{Journal of the Royal Statistical Society. Series B: Statistical
  Methodology}, 69\penalty0 (1):\penalty0 63--78, 2007.
\newblock \doi{10.1111/j.1467-9868.2007.00577.x}.

\bibitem[Wang et~al.(2009)Wang, Li, and Leng]{wangetal2009}
H.~Wang, B.~Li, and C.~Leng.
\newblock {Shrinkage tuning parameter selection with a diverging number of
  parameters}.
\newblock \emph{Journal of the Royal Statistical Society. Series B: Statistical
  Methodology}, 71\penalty0 (3):\penalty0 671--683, 2009.
\newblock ISSN 13697412.
\newblock \doi{10.1111/j.1467-9868.2008.00693.x}.

\bibitem[Yuan and Lin(2006)]{yuanlin2006}
M.~Yuan and Y.~Lin.
\newblock Model selection and estimation in regression with grouped variables.
\newblock \emph{Journal of the Royal Statistical Society: Series B (Statistical
  Methodology)}, 68\penalty0 (1):\penalty0 49--67, 2006.

\bibitem[Yuan and Lin(2007)]{yuanlin2007}
M.~Yuan and Y.~Lin.
\newblock Model selection and estimation in the gaussian graphical model.
\newblock \emph{Biometrika}, 94\penalty0 (1):\penalty0 19--35, 2007.

\bibitem[Zhang and Shen(2010)]{zhangshen2010}
Y.~Zhang and X.~Shen.
\newblock Model selection procedure for high-dimensional data.
\newblock \emph{Statistical Analysis and Data Mining: The ASA Data Science
  Journal}, 3\penalty0 (5):\penalty0 350--358, 2010.

\bibitem[Zhang et~al.(2010)Zhang, Li, and Tsai]{zhangetal2010}
Y.~Zhang, R.~Li, and C.-L. Tsai.
\newblock {Regularization Parameter Selections via Generalized Information
  Criterion}.
\newblock \emph{Journal of the American Statistical Association}, 105\penalty0
  (489):\penalty0 312--323, 2010.

\bibitem[Zou(2006)]{zou2006}
H.~Zou.
\newblock The adaptive lasso and its oracle properties.
\newblock \emph{Journal of the American statistical association}, 101\penalty0
  (476):\penalty0 1418--1429, 2006.

\end{thebibliography}

\end{document}